\documentclass[conference,letterpaper]{IEEEtran}
\usepackage{cite}
\ifCLASSINFOpdf
\else
\fi
\usepackage[cmex10]{amsmath}

\usepackage{amsthm}
\usepackage{amssymb}
\usepackage{bm}
\usepackage[final]{graphicx}
\usepackage{amsfonts}

\usepackage{subfigure}
\usepackage{caption}

\newtheorem{lemma}{Lemma}
\newtheorem{theorem}{Theorem}

\newtheorem{cor}{Corollary}

\begin{document}
%
\title{Connecting Multiple-unicast and Network Error Correction: Reduction and Unachievability}

\author{\IEEEauthorblockN{Wentao Huang}
\IEEEauthorblockA{
California Institute of Technology
}
\and
\IEEEauthorblockN{Michael Langberg}
\IEEEauthorblockA{
University at Buffalo, SUNY
}
\and
\IEEEauthorblockN{Joerg Kliewer}
\IEEEauthorblockA{
New Jersey Institute of Technology
}}


%


\maketitle



\begin{abstract}
We show that solving a multiple-unicast network coding problem can be reduced to solving a single-unicast network error correction problem, where an adversary may jam at most a single edge in the network. Specifically, we present an efficient reduction that maps a multiple-unicast network coding instance to a network error correction instance while preserving feasibility. The reduction holds for both the zero probability of error model and the vanishing probability of error model. Previous reductions are restricted to the zero-error case. As an application of the reduction, we present a constructive example showing that the single-unicast network error correction capacity may not be achievable, a result of separate interest. \end{abstract}

%
\IEEEpeerreviewmaketitle

\allowdisplaybreaks[2]
\section{Introduction}
\footnotetext[0]{This work has been supported in part by NSF grant CCF-1440014, CCF-1440001, CCF-1439465, and CCF-1321129.}
Consider the problem that a source wishes to reliably communicate to a terminal over a network with point-to-point noiseless channels, in the presence of an adversary. The adversary is characterized by a collection $\mathcal{A}$ of subsets of channels, so that it may choose an arbitrary $A \in \mathcal{A}$ and controls the channels in set $A$. Under the assumption that 1) all channels have uniform capacity and 2) $\mathcal{A}$ is the collection of all subsets containing $z$ channels, Yeung and Cai  \cite{Yeung:2006ut} show that the cut-set bound characterizes the network error correction capacity. Efficient capacity-achieving network error correction codes under this setting are proposed in \cite{Zhang:2008wf,Jaggi:2008dq,Kschischang:2008jj,Koetter:2008jt}.

In the settings that either channel capacities are not uniform or $\mathcal{A}$ is arbitrary, determining the network capacity remains an open problem. For both cases it is shown that linear codes are not sufficient to achieve capacity \cite{Kim:2011ec, Kosut:2009ts}. Capacity bounds and achievable strategies for network error correction with unequal channel capacities are studies in \cite{Kim:2011ec}. Achievable strategies for network error correction with non-uniform $\mathcal{A}$, i.e., $\mathcal{A}$ includes subsets of different sizes, are studied in \cite{Kosut:2010ti, DaWang:2010wv, Che:2013vy}. 

The single-source single-terminal network error correction problem with arbitrary $\mathcal{A}$ is shown in a previous work of the authors \cite{Huang:2014vg} to be at least as hard as the multiple-unicast network coding problem (without adversarial errors), using the following reduction technique. For a general multiple-unicast network coding problem $\mathcal{I}$, a corresponding network error correction problem $\mathcal{I}_c$ can be constructed, so that a rate is feasible with zero error in $\mathcal{I}$ if and only if a corresponding rate is feasible with zero error in $\mathcal{I}_c$. Therefore, the problem of determining the zero-error feasibility of a rate in $\mathcal{I}$ is reduced to the problem of determining the zero-error feasibility of a rate in $\mathcal{I}_c$. Under the model that  a vanishing probability of error is allowed, the connection between the feasibility in $\mathcal{I}$ and $\mathcal{I}_c$ is also studied in  \cite{Huang:2014vg}. However, in this case the result therein is weaker, and has a gap so that the ``if and only if'' connection between $\mathcal{I}$ and $\mathcal{I}_c$ is broken. Hence for this case the reduction from $\mathcal{I}$ to $\mathcal{I}_c$ is \emph{not} established. 

In this paper, we close this gap and complete the reduction from multiple-unicast to network error correction under the vanishing-error model.  For a general multiple-unicast network coding problem $\mathcal{I}$, we show that a corresponding network error correction problem $\mathcal{I}_c$ can be constructed, so that a rate is feasible with vanishing error in $\mathcal{I}$ if and only if a corresponding rate is feasible with vanishing error in $\mathcal{I}_c$. We construct $\mathcal{I}_c$ in the same way as in \cite{Huang:2014vg}. 
However, compared to the (implicit) information-theoretic approach used in \cite{Huang:2014vg}, we present a way to explicitly construct the network code for $\mathcal{I}$ from the network code for $\mathcal{I}_c$. 
Furthermore, the new approach enables a stronger result, simplifies the proofs and generalizes the result of the zero-error model as a special case.


As there are connections between $\mathcal{I}$ and $\mathcal{I}_c$ for both zero-error feasibility and vanishing-error feasibility, it is natural to ask if the connection extends to the case of asymptotic feasibility. We answer this question negatively by constructing a counter-example. By applying our analysis to this example, we further show that the (single-source single-terminal) network error correction capacity is not achievable in general, which is a result of separate interest. Previous works \cite{Dougherty:2006cg, Rai:2012iq} have studied the unachievability of capacity for multiple-unicast networks and sum-networks, respectively.
\section{Models}
\subsection{Multiple-unicast Network Coding}
A network is a directed graph $\mathcal{G}=(\mathcal{V},\mathcal{E})$, where vertices represent network nodes and edges  represent  channels. Each edge $e \in \mathcal{E}$ has a capacity $c_e$, which is the number of bits 
 that can be transmitted on $e$ in one transmission. An instance  $\mathcal{I}=(\mathcal{G}, \mathcal{S}, \mathcal{T}, B)$ of the  \emph{multiple-unicast network coding problem},  includes a network $\mathcal{G}$, a set of source nodes $\mathcal{S} \subset \mathcal{V}$, a set of terminal nodes $\mathcal{T} \subset \mathcal{V}$ and an $|\mathcal{S}|$ by $|\mathcal{T}|$ requirement matrix $B$. The $(i,j)$-th entry of $B$ equals 1 if terminal $j$ requires the information from source $i$ and equals 0 otherwise.  $B$ is assumed to be a permutation matrix so that each source is paired with a single terminal. Denote by $s(t)$ the source required by terminal $t$. Denote $[n] \triangleq \{1,.., n \}$, and each source $s \in \mathcal{S}$ is associated with an independent  message, represented by a random variable $M_s$ uniformly distributed over $[2^{nR_s}]$. 
 A \emph{network code} of length $n$ is a set of encoding functions $\phi_e$ for every $e \in \mathcal{E}$ and a set of decoding functions $\phi_t$ for each $t \in \mathcal{T}$. For each $e =(u,v) $, the encoding function $\phi_e$ is a function taking as input the signals received from the incoming edges of node $u$, as well as the random variable $M_u$ if $u \in \mathcal{S}$. $\phi_e$ evaluates to a value in $[2^{n c_e}]$, which is the signal transmitted on $e$. For each $t \in \mathcal{T}$, the decoding function $\phi_t$ maps the tuple of signals received from the incoming edges of $t$, to an estimated message $\hat{M}_{s(t)}$ with values in  $[2^{nR_{s(t)}}]$.

A network code $\{\phi_e, \phi_t \}_{e \in \mathcal{E}, t \in \mathcal{T}}$ is said to \emph{satisfy} a terminal $t$ under transmission $(m_s, s\in\mathcal{S})$ if $\hat{M}_{s(t)} = m_{s(t)}$ when $(M_s, s \in \mathcal{S}) = (m_s, s \in \mathcal{S})$. A network code is said to satisfy the multiple-unicast network coding problem $\mathcal{I}$ with error probability $\epsilon$ if the probability that all $t \in \mathcal{T}$ are simultaneously  satisfied is at least $1-\epsilon$. The probability is taken over the joint distribution on $(M_s, s\in\mathcal{S})$. Formally, the network code satisfies $\mathcal{I}$ with error probability $\epsilon$ if
\begin{align*}
\Pr_{(M_s, s \in \mathcal{S})} \left\{ \bigcap_{t \in \mathcal{T}} t \text{ is satisfied under }(M_s, s \in \mathcal{S})  \right\} \ge 1-\epsilon.
\end{align*}

For an instance $\mathcal{I}$ of the multiple-unicast network coding problem, rate $R$ is said to be \emph{feasible} if $R_s = R$, $\forall s\in\mathcal{S}$, and for any $\epsilon>0$, there exists a network code that satisfies $\mathcal{I}$ with error probability at most $\epsilon$. Rate $R$ is said to be \emph{feasible with zero error} if $R_s = R$, $\forall s\in\mathcal{S}$ and there exists a network code that satisfies $\mathcal{I}$ with zero error probability. Rate $R$ is said to be \emph{asymptotically feasible} if  for any $\delta > 0$, rate $(1-\delta)R$ is feasible. 
The model assumes that all sources transmit information at equal rate. There is no loss of generality in this assumption as a varying rate source $s$ can be modeled by several equal rate sources co-located at $s$.

\subsection{Single-source Single-terminal Network Error Correction}
An instance $\mathcal{I}_c = (\mathcal{G}, s, t, \mathcal{A})$ of the \emph{single-source single-terminal network error correction problem} includes a network $\mathcal{G}$, a source node $s$, a terminal node $t$ and a collection of subsets of channels $\mathcal{A} \subset 2^\mathcal{E}$ susceptible to errors. An error occurs in a channel if the output of the channel is different from the input. More precisely, the output of a channel $e$ is the bitwise xor of the input signal and an error signal $r_e$. We say there is an error in channel $e$ if $r_e$ is not the zero vector. For a subset $A \in \mathcal{A}$ of channels,  an $A$-error is said to occur if an error occurs in every channel in $A$. Since there is only a single source and a single terminal, in this problem we suppress the subscript and denote by $M$ the random message of the source, by $R$ the rate of $M$, and by $\hat{M}$ the output of the decoder at $t$.

Denote by $\mathcal{R}_{\mathcal{A}}$ the set of \emph{error patterns} $\bm{r} = \{r_e\}_{e \in \mathcal{E}}$ that correspond to an $A$-error, for any $A \in \mathcal{A}$. A network code $\{\phi_e, \phi_t\}_{e \in \mathcal{E} }$ is said to \emph{satisfy} $\mathcal{I}_c$ under transmission $m$ if $\hat{M} = m$ when $M = m$, regardless of the occurrence of any error pattern $\bm{r} \in \mathcal{R}_{\mathcal{A}}$. A network code is said to satisfy problem $\mathcal{I}_c$ with error probability $\epsilon$ if the probability that $\mathcal{I}_c$ is  satisfied is at least $1-\epsilon$. The probability is taken over the  distribution on the source message $M$. Note that our model targets the worst-case scenario in the sense that if $M=m$ is transmitted and $m$ is satisfied by the network code, then correct decoding is guaranteed regardless of the error pattern.

For a single-source single-terminal network error correction problem $\mathcal{I}_c$, rate $R$ is said to be \emph{feasible} if for any $\epsilon>0$, there exists a network code that satisfies $\mathcal{I}_c$ with error probability at most $\epsilon$. Rate $R$ is said to be \emph{feasible with zero error} if there exists a network code that satisfies $\mathcal{I}_c$ with zero error probability. Rate $R$ is said to be \emph{asymptotically feasible} if  for any $\delta > 0$, rate $(1-\delta)R$ is feasible. The capacity of $\mathcal{I}_c$ is the supremum over all rates that are asymptotically feasible.


\section{Reduction from Multiple-unicast to Network Error Correction}\label{sec:red}
\begin{figure}[h!]
  \begin{center}
      \includegraphics[width=0.27\textwidth]{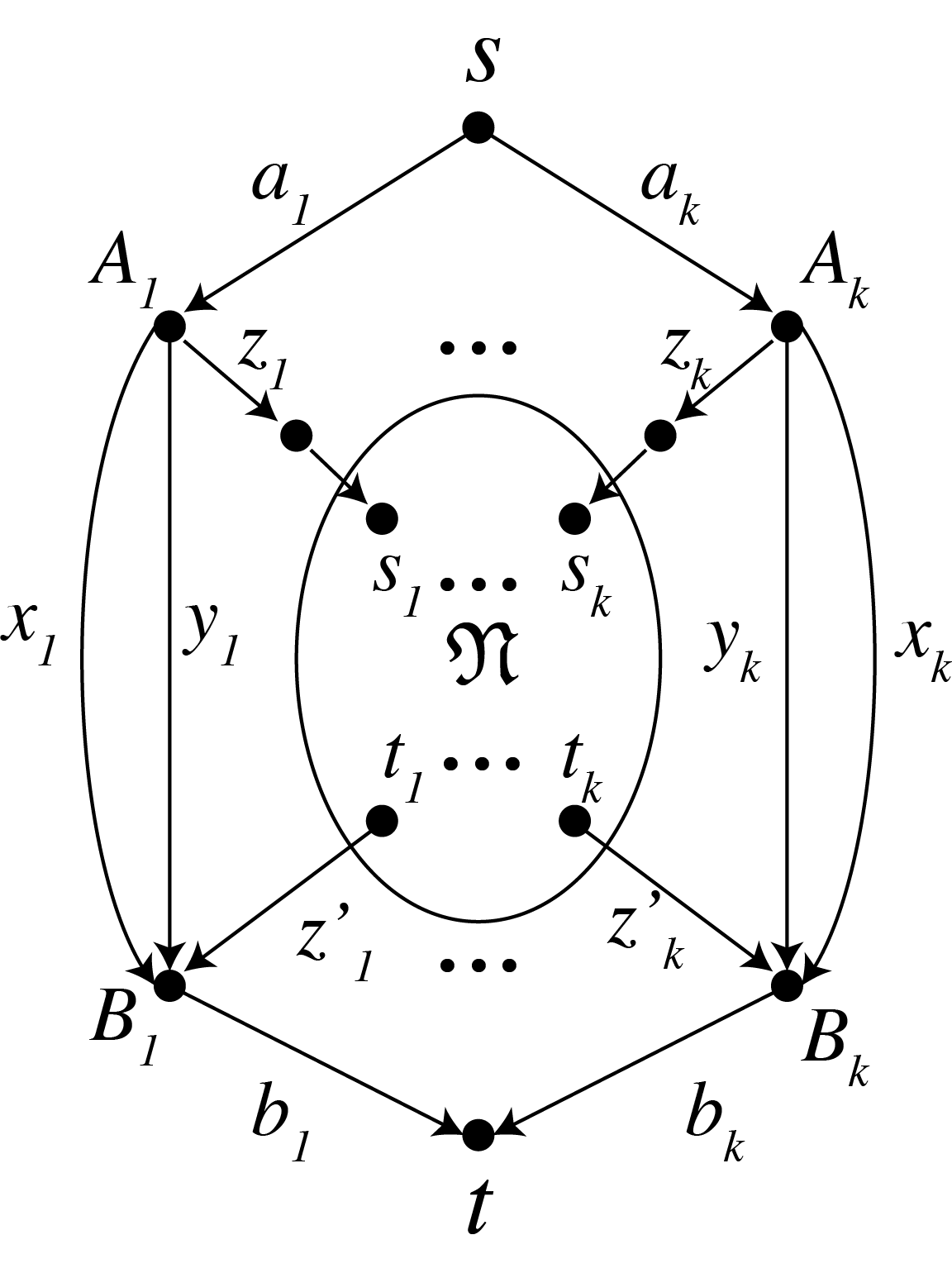}
  \caption{In the single-unicast network error correction problem $\mathcal{I}_c$, the source $s$ wants to communicate with the terminal $t$. $\mathcal{N}$ is a general network with point-to-point noiseless channels. All edges outside $\mathcal{N}$ (i.e., edges for which at least one of its end-points does not belong to $\mathcal{N}$) have unit capacity. There is at most one error in this network, and this error may occur at any edge except $\{ a_i, b_i, 1 \le i \le k \}$. Namely, $\mathcal{A}$ includes all singleton sets of a single edge in the network except $\{a_i\}$ and $\{b_i\}$, $i=1,...,k$. Note that there are $k$ parallel branches in total but only the first and the $k$-th branches are drawn explicitly. The multiple-unicast network coding problem $\mathcal{I}$ is defined on the network $\mathcal{N}$, where the $k$ source-destination pairs are $(s_i,t_i), i=1,...,k$, and all channels are error-free.}\label{zeroerr}
           \end{center}
\end{figure}
The following theorem is the main result of this section.
\begin{theorem}\label{epsilonerrk2}
Given any multiple-unicast network coding problem $\mathcal{I}$ with source-destination pairs $\{ (s_i, t_i), i=1,...,k \}$, a corresponding single-source single-terminal network error correction problem $\mathcal{I}_c=(\mathcal{G},s,t,\mathcal{A})$  in which $\mathcal{A}$ includes sets with at most a single edge can be constructed as specified in Figure \ref{zeroerr},  such that rate $k$ is feasible in $\mathcal{I}_c$ if and only if unit rate is feasible in $\mathcal{I}$.
\end{theorem}
The backward direction of the theorem, i.e., unit rate in $\mathcal{I}$ implies rate $k$ in $\mathcal{I}_c$, is simple and is prove in \cite{Huang:2014vg}. 
In the remainder of this section we prove the forward direction of Theorem \ref{epsilonerrk2}, i.e., the feasibility of rate $k$ in $\mathcal{I}_c$ implies the feasibility of unit rate in $\mathcal{I}$.

Suppose in $\mathcal{I}_c$ a rate of $k$ is achieved by a network code $\mathcal{C}=\{ \phi_e, \phi_t\}_{e\in\mathcal{E}}$ with length $n$, and with a probability of error $\epsilon$. Recall that $M$ is the source message uniformly distributed over $\mathcal{M} = [2^{kn}]$, and $\hat{M}$ is the output of the decoder at the terminal.  Let $\mathcal{M}^{\text{good}}=\{ m \in [2^{kn}] : \mathcal{C} \text{ satisfies } \mathcal{I}_c \text{ under transmission } m \}$ be the subset of messages that can be decoded correctly under any error pattern $\bm{r} \in \mathcal{R}_{\mathcal{A}}$. Denote by $\mathcal{M}^{\text{bad}} = \mathcal{M} \backslash \mathcal{M}^{\text{good}}$, then for any $m \in \mathcal{M}^{\text{bad}}$, there exists an error pattern $\bm{r} \in \mathcal{R}_{\mathcal{A}}$ such that $\hat{M} \ne m$ if $M=m$ and $\bm{r}$ occurs, i.e., a decoding error occurs. Because $\mathcal{C}$ satisfies $\mathcal{I}_c$ with error probability $\epsilon$, it follows that $|\mathcal{M}^{\text{bad}}| \le 2^{kn} \epsilon$ and thus $|\mathcal{M}^{\text{good}}| \ge (1-\epsilon) \cdot 2^{kn}$.

We introduce some notation needed in the proof. In problem $\mathcal{I}_c$, under the network code $\mathcal{C}$, for $i=1,...,k$, let $x_i(m,\bm{r}): \mathcal{M} \times \mathcal{R}_{\mathcal{A}} \to [2^n]$ be the  signal received from channel $x_i$ when $M=m$ and the error pattern $\bm{r}$ happens. Let $\bm{r} = \bm{0}$ denotes the case that no error has occurred  in the network. Let $x_i(m) = x_i(m,\bm{0})$, $\bm{x}(m,\bm{r}) = (x_1(m,\bm{r}) , ...,  x_k(m,\bm{r}) )$ and $\bm{x}(m) = (x_1(m), ...,  x_k(m))$. We define functions $a_i, b_i, y_i, z_i, z'_i, \bm{a}, \bm{b}, \bm{y}, \bm{z}, \bm{z}'$ for problem $\mathcal{I}_c$ in a similar way.

Notice that the set of edges $a_1, ..., a_k$ forms a cut-set from $s$ to $t$, and so does the set of edges $b_1, ..., b_k$.
Therefore for any $m_1, m_2 \in \mathcal{M}^{\text{good}}$, $m_1 \ne m_2$, it follows from the decodability constraint that $\bm{a}(m_1) \ne \bm{a}(m_2)$ and $\bm{b}(m_1) \ne \bm{b}(m_2)$. Setting $\mathcal{B}^{\text{good}} = \{ \bm{b}(m) : m \in \mathcal{M}^{\text{good}} \}$, it then follows from   $|\mathcal{M}^{\text{good}}| \ge (1-\epsilon) \cdot 2^{kn}$ that $|\mathcal{B}^{\text{good}}| \ge (1-\epsilon) \cdot 2^{kn} $. Setting $\mathcal{B}^\text{err} = [2^n]^k \backslash \mathcal{B}^\text{good}$, it follows that $|\mathcal{B}^\text{err}| \le  2^{kn} \epsilon$. We define $\mathcal{A}^\text{good}$ and $\mathcal{A}^\text{err}$ similarly and so $|\mathcal{A}^{\text{good}}| \ge (1-\epsilon) \cdot 2^{kn}$, $|\mathcal{A}^\text{err}| \le  2^{kn} \epsilon$.

Let $\mathcal{M}(\hat{z}'_i, \hat{b}_i) = \{ m \in \mathcal{M}^\text{good}: z'_i(m) = \hat{z}'_i, b_i(m) = \hat{b}_i  \}$, we define a  function $\psi_i : [2^n] \to [2^n]$  as:
\begin{align}\label{psii}
\psi_i(\hat{z}'_i) = \arg \max_{\hat{b}_i  } | \mathcal{M}(\hat{z}'_i, \hat{b}_i) | \triangleq \hat{b}_{i, \hat{z}'_i}
\end{align}
Function $\psi_i$ will be useful later, when we design the network codes in $\mathcal{I}$. Intuitively, in the absence of adversarial errors,  $\psi_i$ estimates the signal transmitted on edge $b_i$ given that the signal transmitted on edge $z'_i$ is $\hat{z}'_i$. In the following we analyze how often will $\psi_i$ make a mistake. Define $\mathcal{M}^\psi_i = \{ m \in  \mathcal{M}^\text{good} : \psi_i(z'_i(m)) \ne b_i(m) \}$. Notice that $\mathcal{M}_i^\psi$ is the set of messages that, when they are transmitted by the source, $\psi_i$ will make a mistake in guessing the signal transmitted on $b_i$. Lemma \ref{psierr} shows that the size of this set is small. 
\begin{lemma}\label{psierr}
$|\mathcal{M}^\psi_i| \le 2 \epsilon \cdot 2^{kn}$.
\end{lemma}
We make the following combinatorial observation before proving Lemma~\ref{psierr}. Lemma~\ref{lfour} is a variation of \cite[Lemma 4]{Huang:2014vg}

\begin{lemma}\label{lfour}
Let $\mathcal{M}(\hat{z}'_i) = \{m \in \mathcal{M}^\text{good}: z'_i(m) = \hat{z}'_i\}$, then for any  $m_1, m_2 \in \mathcal{M}(\hat{z}'_i)$ such that $b_i(m_1) \ne b_i(m_2)$, there exists an element of $\mathcal{B}^\text{err}$ that will be decoded by terminal $t$ to either $m_1$ or $m_2$.
\end{lemma}
\begin{proof}
Consider any $m_1, m_2 \in \mathcal{M}(\hat{z}'_i)$ such that $b_i(m_1) \ne b_i(m_2)$. Let $\bm{r}_1$ be the error pattern that changes the signal on $x_i$ to be $x_i(m_2)$, and let $\bm{r}_2$ be the error pattern that changes the signal on $y_i$ to be $y_i(m_1)$. Then if $m_1$ is transmitted by the source and $\bm{r}_1$ happens, node $B_i$ will receive the same inputs $(x_i(m_2), y_i(m_1), z'_i(m_1)=z'_i(m_2))$ as in the situation that $m_2$ is transmitted and $\bm{r}_2$ happens. Therefore $b_i(m_1, \bm{r}_1) = b_i(m_2, \bm{r}_2)$, and so either $b_i(m_1, \bm{r}_1)  \ne b_i(m_1)$ or $b_i(m_2, \bm{r}_2) \ne b_i(m_2)$ because by hypothesis $b_i(m_1) \ne b_i(m_2)$. Consider the first case that $b_i(m_1, \bm{r}_1)  \ne b_i(m_1)$, then the tuple of signals $(b_1(m_1, \bm{r}_1), ..., b_k(m_1, \bm{r}_1)) = (b_1(m_1), ..., b_i(m_1, \bm{r}_1), ..., b_k(m_1))$ will be decoded by the terminal to message $m_1$ because of the fact that $m_1 \in \mathcal{M}^\text{good}$ which is correctly decodable under any error pattern $\bm{r} \in \mathcal{R}_\mathcal{A}$. Therefore this tuple of signals is an element of $\mathcal{B}^\text{err}$ since it does not equals $\bm{b}(m_1)=(b_1(m_1),...,b_k(m_1))$ and it does not equal $\bm{b}(m)$, for any $m \ne m_1$, $m \in \mathcal{M}^\text{good}$, because otherwise it will be decoded by the terminal to $m$.
Similarly in the latter case that $b_1(m_2, \bm{r}_2) \ne b_1(m_2)$, then $(b_2(m_2, \bm{r}_2), ..., b_k(m_2, \bm{r}_2)) = (b_2(m_2), ..., b_i(m_2, \bm{r}_2), ..., b_k(m_2))$ is an element of $\mathcal{B}^\text{err}$ and will be decoded by the terminal to $m_2$.
Therefore in both cases we are able to find an element of $\mathcal{B}^\text{err}$ that will be decoded by the terminal to either $m_1$ or $m_2$.
\end{proof}


\begin{proof}[Proof (of Lemma \ref{psierr})]
We can partition $\mathcal{M}^\psi_i$ as
\begin{align*}
\mathcal{M}^\psi_i = \bigcup_{\hat{z}'_i} \left( \mathcal{M}(\hat{z}'_i) \backslash \mathcal{M}(\hat{z}'_i,\hat{b}_{i,\hat{z}'_i}) \right),
\end{align*}
and so
\begin{align}\label{mpsi}
|\mathcal{M}^\psi_i | = \sum_{\hat{z}'_i} \left( |\mathcal{M}(\hat{z}'_i)| - |\mathcal{M}(\hat{z}'_i,\hat{b}_{i,\hat{z}'_i})| \right).
\end{align}

Consider an arbitrary $\hat{z}'_i$ and the set $\mathcal{M}(\hat{z}'_i)$. We define an iterative procedure as follows. Initialize $\mathcal{W} := \mathcal{M}(\hat{z}'_i)$. If there exist two messages $m_1, m_2 \in \mathcal{W}$ such that $b_i(m_1) \ne b_i(m_2)$, then delete both $m_1, m_2$  from $\mathcal{W}$. Repeat the operation until there does not exist $m_1, m_2 \in \mathcal{W}$ such that $b_i(m_1) \ne b_i(m_2)$.

After the procedure terminates, it follows that $|\mathcal{W}| \le  |\mathcal{M}(\hat{z}'_i,\hat{b}_{i,\hat{z}'_i})|$, otherwise by the definition of $\hat{b}_{i,\hat{z}'_i}$ there must exist $m_1, m_2 \in \mathcal{W}$ such that $b_i(m_1) \ne b_i(m_2)$. Therefore at least $|\mathcal{M}(\hat{z}'_i)| - |\mathcal{M}(\hat{z}'_i,\hat{b}_{i,\hat{z}'_i})|$  elements are deleted from $\mathcal{M}(\hat{z}'_i)$. By Lemma \ref{lfour}, each pair of elements deleted corresponds to an element of $\mathcal{B}^\text{err}$. Also by Lemma \ref{lfour} the elements of $\mathcal{B}^\text{err}$ corresponding to different deleted pairs are distinct. Summing over all possible values of $\hat{z}'_i$, it follows that the total number of deleted pairs is smaller than the size of $\mathcal{B}^\text{err}$:
\begin{align}
\frac{1}{2} \sum_{\hat{z}'_i} \left( |\mathcal{M}(\hat{z}'_i)| - |\mathcal{M}(\hat{z}'_i,\hat{b}_{i,\hat{z}'_i})| \right) \nonumber &\\
& \hspace{-23mm} \le \sum_{\hat{z}'_i} \text{ \# of pairs deleted from }\mathcal{M}(\hat{z}'_i) \nonumber \\
& \hspace{-23mm} \le |\mathcal{B}^\text{err}| = \epsilon \cdot 2^{kn}.\label{ndel}
\end{align}

Combining (\ref{mpsi}) and (\ref{ndel}) we have $|\mathcal{M}^\psi_i| \le 2 \epsilon \cdot 2^{kn}$.
\end{proof}

Next, let $\mathcal{M}(\hat{a}_i, \hat{b}_i) = \{ m \in \mathcal{M}^\text{good}: a_i(m) = \hat{a}_i, b_i(m) = \hat{b}_i  \}$, we define a  function $\pi_i : [2^n] \to [2^n]$  as:
\begin{align}\label{pii}
\pi_i(\hat{b}_i) = \arg \max_{\hat{a}_i  } | \mathcal{M}(\hat{a}_i, \hat{b}_i) | \triangleq \hat{a}_{i, \hat{b}_i}
\end{align}

Function $\pi_i$ will be useful later for designing the network codes in $\mathcal{I}$. Intuitively, in the absence of adversarial errors, $\pi_i$ estimates the signal transmitted on edge $a_i$ given that the signal transmitted on edge $b_i$ is $\hat{b}_i$.
In the following we analyze how often will $\pi_i$ make a mistake. Define $\mathcal{M}^\pi_i = \{ m \in  \mathcal{M}^\text{good} : \pi_i(b_i(m)) \ne a_i(m) \}$.  Notice that $\mathcal{M}_i^\pi$ is the set of messages that, when they are transmitted by the source, $\pi_i$ will make a mistake in guessing the signal transmitted on $a_i$.  Lemma \ref{pierr} shows that the size of this set is small. 

\begin{lemma}\label{pierr}
$|\mathcal{M}^\pi_i| \le 3 \epsilon \cdot 2^{kn}$.
\end{lemma}
 We make the following combinatorial observation before proving Lemma~\ref{pierr}. Lemma \ref{three} is a variation of \cite[Lemma 3]{Huang:2014vg}.
\begin{lemma}\label{three}
 Define $\mathcal{M}(\hat{a}_i) = \{ m \in \mathcal{M}^\text{good} : a_i(m) = \hat{a}_i \}$. If $|\{b_i(m):m \in \mathcal{M}(\hat{a}_i)\}| = L $, then there exist $(L-1)| \mathcal{M} (\hat{a}_i)|$ distinct elements of $\mathcal{B}^\text{err}$ such that each of them will be decoded by terminal $t$ to some message $m \in \mathcal{M}(\hat{a}_i)$.
\end{lemma}
\begin{proof}
Assume for concreteness that $\{b_i(m):m \in \mathcal{M}(\hat{a}_i)\} = \{ \hat{b}^{(1)}_i , ..., \hat{b}_i^{(L)} \}$, then there exist $L$ messages $m_1, ..., m_L \in \mathcal{M}(\hat{a}_i)$ such that $b_i(m_j) = \hat{b}_i^{(j)}$, $j=1,...,L$. For $j =1,...,L$, let $\bm{r}_j$ be the error pattern that changes the signal on $z'_i$ to be $z'_i(m_j)$. Then if a message $m_0 \in \mathcal{M}(\hat{a}_i)$ is transmitted by the source and $\bm{r}_j$ happens, the node $B_i$ will receive the same inputs $(x_i(m_0), y_i(m_0), z'_i(m_j))$ as in the situation that $m_j$ is sent and no error happens. Therefore $b_i(m_0, \bm{r}_j) = \hat{b}_i^{(j)}$, and so $|\{ \bm{b}(m_0, \bm{r}_j) \}_{j\in [L]}| = | \{ b_i(m_0,\bm{r}_j) \}_{j \in [L]} | = L$. Since $m_0 \in \mathcal{M}^\text{good}$, it is correctly decodable under any error pattern $\bm{r} \in \mathcal{R}_{\mathcal{A}}$, and so all elements of $\{ \bm{b}(m_0, \bm{r}_j) \}_{j\in [L]}$ will be decoded by the terminal to $m_0$. Except the element $\bm{b}(m_0)$, the other $L-1$ elements of $\{ \bm{b}(m_0, \bm{r}_j) \}_{j\in [L]}$ are elements of $\mathcal{B}^\text{err}$. Sum over all $ m_0 \in \mathcal{M}(\hat{a}_i)$ and the assertion is proved.
\end{proof}

\begin{proof}[Proof (of Lemma \ref{pierr})]
Define $\mathcal{A}_{i,1}^\pi = \{ \hat{a}_i \in [2^n] : |\mathcal{M}(\hat{a}_i)| \le \frac{1}{2} 2^{(k-1)n} \} $, and $\mathcal{A}_{i,2}^\pi =\{ \hat{a}_i \in [2^n] \backslash \mathcal{A}_{i,1}^\pi   : |\{b_i(m):m \in \mathcal{M}(\hat{a}_i)\}| > 1   \}$. Then define $\mathcal{M}_{i,1}^\pi = \{ m \in \mathcal{M}^\text{good} : a_i(m) \in \mathcal{A}_{i,1}^\pi   \}$, and $\mathcal{M}_{i,2}^\pi =  \{ m \in \mathcal{M}^\text{good} : a_i(m) \in  \mathcal{A}_{i,2}^\pi   \}$. Notice that by construction $\mathcal{A}_{i,1}^\pi$ and $\mathcal{A}_{i,2}^\pi$ are disjoint, and $\mathcal{M}_{i,1}^\pi $ and $\mathcal{M}_{i,2}^\pi $ are disjoint. We claim that,
\begin{align}
\mathcal{M}^\pi_i \subset \mathcal{M}_{i,1}^\pi \cup \mathcal{M}_{i,2}^\pi.\label{mpi}
\end{align}
To prove the claim, consider any $m \in \mathcal{M}^\text{good}$ such that $m \notin \mathcal{M}_{i,1}^\pi \cup \mathcal{M}_{i,2}^\pi$. We will show that $\pi( b_i(m) ) = a_i(m)$. Suppose for the sake of contradiction that $\pi( b_i(m) ) = \hat{a}_i \ne a_i(m)$, 
then it follows that
\begin{align}
| \mathcal{M}(\hat{a}_i, b_i(m)) | \stackrel{(a)}{>} |\mathcal{M}(a_i(m), b_i(m))|& \nonumber \\  & \hspace{-18mm} \stackrel{(b)}{=} |\mathcal{M}(a_i(m))|  \stackrel{(c)}{>} \frac{1}{2} 2^{(k-1)n},
\end{align}
where (a) is due to the definition of $\pi$, (b) is due to the fact that $m \notin \mathcal{M}_{i,2}^\pi$ and (c) is due to the fact that $m \notin \mathcal{M}_{i,1}^\pi$. Let $\mathcal{M}(\hat{b}_i) = \{ m' \in \mathcal{M}^\text{good} : b_i(m') = \hat{b}_i \}$, then $ \mathcal{M}(\hat{a}_i, b_i(m)) \cup \mathcal{M}(a_i(m)) \subset \mathcal{M}(b_i(m))$. Since $\hat{a}_i \ne a_i(m)$, $ \mathcal{M}(\hat{a}_i, b_i(m))$ and $\mathcal{M}(a_i(m))$ are disjoint, and it follows that $| \mathcal{M}(b_i(m))| \ge |\mathcal{M}(\hat{a}_i, b_i(m))| + |\mathcal{M}(a_i(m))| > 2^{(k-1)n}$. However, because $|\{ (\hat{b}_1, ..., \hat{b}_k ) \in [2^n]^k  : \hat{b}_i = b_i(m) \}| = 2^{(k-1)n}$, by the pigeonhole principle there must exist two messages $m_1, m_2 \in \mathcal{M}(b_i(m))$ such that $\bm{b}(m_1) = \bm{b}(m_2)$. This is a contradiction since the terminal cannot distinguish $m_1$ from $m_2$. This proves $\pi( b_i(m) ) = a_i(m)$ as well as (\ref{mpi}).

We next bound the size of $\mathcal{M}_{i,1}^\pi$ and $\mathcal{M}_{i,2}^\pi$. 
For any $\hat{a}'_i \in \mathcal{A}_{i,1}^\pi$, by defnition  $\{  (\hat{a}_1 ,..., \hat{a}_k) \in [2^n]^k : \hat{a}_i = \hat{a}'_i   \} \backslash \{ (\bm{a}(m) : m \in \mathcal{M}(\hat{a}'_i)  \}$ is a subset of $\mathcal{A}^\text{err}$ with size at least $\frac{1}{2} 2^{(k-1)n}$. Therefore each element of $\mathcal{A}_{i,1}^\pi$ will contribute to at least $\frac{1}{2} 2^{(k-1)n}$ distinct elements of $\mathcal{A}^\text{err}$. Hence $|\mathcal{A}_{i,1}^\pi| \cdot \frac{1}{2} 2^{(k-1)n} \le |\mathcal{A}^\text{err}| = \epsilon \cdot 2^{kn}$, and so $|\mathcal{A}_{i,1}^\pi| \le 2\epsilon \cdot 2^n$. It then follows that $|\mathcal{M}_{i,1}^\pi| \le \frac{1}{2} 2^{(k-1)n} |\mathcal{A}_{i,1}^\pi| = \epsilon \cdot 2^{kn}$.

By Lemma \ref{three}, each elements of $\mathcal{A}_{i,2}^\pi$ will contribute to at least $ \frac{1}{2} 2^{(k-1)n} $ distinct elements in $\mathcal{B}^\text{err}$. Therefore $|\mathcal{A}_{i,2}^\pi| \cdot \frac{1}{2} 2^{(k-1)n} \le |\mathcal{B}^\text{err}| = \epsilon \cdot 2^{kn}$, and so $|\mathcal{A}_{i,2}^\pi| \le 2\epsilon \cdot 2^n$. It then follows that $|\mathcal{M}_{i,2}^\pi| \le 2^{(k-1)n} |\mathcal{A}_{i,2}^\pi| = 2\epsilon \cdot 2^{kn}$. Finally, by (\ref{mpi}) we have $|\mathcal{M}^\pi_i| \le |\mathcal{M}^\pi_{i,1}| + |\mathcal{M}^\pi_{i,2}| \le 3 \epsilon \cdot 2^{kn}$.
\end{proof}

We are now ready to prove Theorem \ref{epsilonerrk2}.
\begin{proof}[Proof (``$\Rightarrow$'' part of Theorem \ref{epsilonerrk2})]
We show the feasibility of rate $k$ in $\mathcal{I}_c$ implies the feasibility of unit rate in $\mathcal{I}$.

Let $\{\phi_e, \phi_t \}_{e \in \mathcal{E}}$ be the network error correction code of length $n$ that achieves rate $k$ in $\mathcal{I}_c$, with probability of error $\epsilon$. We assume that in this code edge $z_i$ simply relays the signal from edge $a_i$. This is without loss of generality because for any network code that needs to process the signal on edge $a_i$ to obtain the signal to be transmitted on edge $z_i$, it is equivalent to relay the signal on edge $z_i$ and perform the processing work at the head node of edge $z_i$.

Let $\mathcal{E}_\mathcal{N} \subset \mathcal{E}$ be the set of edges of the embedded graph $\mathcal{N}$. For the multiple-unicast problem $\mathcal{I}$, we define a length-$n$ network code $\{ \tau_e, \tau_{t_i} : e \in \mathcal{E}_\mathcal{N}, i \in [k] \}$ as follows.
\begin{align*}
\tau_e & = \phi_e, \ \ \ \ \forall e \in \mathcal{E}_{\mathcal{N}}\\
\tau_{t_i} & = \phi_{z_i} \circ \pi_i \circ  \psi_i \circ  \phi_{z'_i}, \ \ \ \ \forall i = 1,...,k.
\end{align*}
where $\circ$ denotes function composition; $\phi_{z_i}$ and  $\phi_{z'_i}$ are the encoding functions of edges $z_i$ and $z'_i$ in problem $\mathcal{I}_c$; $\psi_i$ is defined in $(\ref{psii})$; and $\pi_i$ is defined in (\ref{pii}). In the following we show that $\{ \tau_e, \tau_{t_i} : e \in \mathcal{E}_\mathcal{N}, i \in [k] \}$ achieves unit rate in $\mathcal{I}$ with probability of error upper bounded by $6k \epsilon$.

In problem $\mathcal{I}$, let $M_i$ be the random message associated with source $s_i$, then $M_i$, $i=1,...,k$ are i.i.d. uniformly distributed over $[2^n]$. Denote for short $\bm{M} = (M_1, ..., M_k)$, then slightly abusing  notations we denote by $\tau_{t_i}(\bm{M}) $ the output of the decoder $\tau_{t_i}$ under transmission $\bm{M}$.
The probability of decoding error is given by
\begin{align*}
 \Pr  \{ \bigcup_{i=1}^k \tau_{t_i}(\bm{M}) \ne M_i  \},
\end{align*}
where the probability is taken over the joint distribution of the random messages. Let $\bm{m} = (m_1, ..., m_k)$ be the realization of $\bm{M}$. We claim that if there exists a message $m$ of problem $\mathcal{I}_c$ (not to be confused with $\bm{m}$, a message of $\mathcal{I}$) such that $m \in \mathcal{M}^\text{good}$, $m \notin \mathcal{M}_i^{\psi}$, $m \notin \mathcal{M}_i^\pi$ and $\bm{m} = \bm{z}(m)$, then $\tau_{t_i}(\bm{m}) = m_i$. To prove the claim, suppose $\bm{m} = \bm{z}(m)$ is transmitted in $\mathcal{I}$. Notice that all edges in $\mathcal{N}$ perform the same coding scheme in $\mathcal{I}$ as in $\mathcal{I}_c$, therefore for terminal node $t_i$, by invoking the function $\phi_{z'_i}$, it obtains $z'_i(m)$. Then by the definition of $\mathcal{M}_i^{\psi}$, it follows that $\psi_i(z'_i(m)) = b_i(m)$. And by the definition of $\mathcal{M}_i^\pi$, it follows that $\pi(\psi_i({z'_i}(m))) = a_i(m)$. Finally since $\bm{m} = \bm{z}(m)$, it follows that $\phi_{z_i}( \pi(\psi_i({z'_i}(m)))) = \phi_{z_i}( a_i(m) ) = z_i(m) = m_i $.

Therefore $\tau_{t_i}(\bm{m}) = m_i$ if $\bm{m} \in \{   \bm{z}(m) \in [2^n]^k:  m \in \mathcal{M}^\text{good}, m \notin \mathcal{M}^\psi_i, m \notin \mathcal{M}^\pi_i \}$.  The probability that $\tau_{t_i}$ makes an error, i.e., $ \Pr  \{ \tau_{t_i}(\bm{M}) \ne M_i  \}$, is upper bounded by the probability of the union of the following three events.
\begin{align*}
E_1 &= \{ \bm{M} = \bm{m}: \bm{m} \notin \{ \bm{{z}}(m) \in [2^n]^k :  m \in \mathcal{M}^\text{good} \} \}\\
E_2 &= \{ \bm{M} = \bm{m}: \bm{m} \in \{ \bm{{z}}(m) \in [2^n]^k : m \in \mathcal{M}^\psi_i  \} \}\\
E_3 &= \{ \bm{M} = \bm{m}: \bm{m} \in \{ \bm{{z}}(m) \in [2^n]^k : m \in \mathcal{M}^\pi_i  \}\}.
\end{align*}

We upper bound the probability of $E_1, E_2, E_3$, respectively.
\begin{align}
\Pr \{ E_1\}  &= 1- \frac{ | \{ \bm{z}(m): m \in \mathcal{M}^\text{good}  \} | }{2^{kn}}\nonumber\\
&\stackrel{(d)}{=} 1 - \frac{|\mathcal{M}^\text{good}|}{2^{kn}} \le 1 -\frac{(1-\epsilon) \cdot 2^{kn}}{2^{kn}}=\epsilon,\label{e1}
\end{align}
where (d) follows from the fact that $ \bm{z}(m) = \bm{a}(m) \ne \bm{a}(m') = \bm{z}(m')$ for any $m,m' \in \mathcal{M}^\text{good}$, $m \ne m'$.
By Lemma \ref{psierr},
\begin{align}\label{e2}
\Pr \{E_2\} = \frac{|\mathcal{M}^\psi_i|}{2^{kn}} \le 2\epsilon. 
\end{align}
And by Lemma \ref{pierr}, we have
\begin{align}\label{e3}
\Pr \{E_3\} = \frac{|\mathcal{M}^\pi_i|}{2^{kn}} \le 3\epsilon. 
\end{align}

Combining (\ref{e1}), (\ref{e2}) and (\ref{e3}), it follows that
\begin{align*}
 \Pr  \{ \tau_{t_i}(\bm{M}) \ne M_i  \} & \le \Pr \{ E_1\} + \Pr \{ E_2\} +\Pr \{ E_3\} \le 6\epsilon.
\end{align*}
Finally, by taking the union bound over the $k$ terminals,
\begin{align*}
 \Pr   \{ \bigcup_{i=1}^k \tau_{t_i}(\bm{M}) \ne M_i  \} & \le 6 k \epsilon.
\end{align*}
Hence the probability of error is arbitrarily small and this establishes the feasibility of unit rate in $\mathcal{I}$. 
\end{proof}

The proof above suggests that the feasibility of rate $k$ with error probability $\epsilon$ in $\mathcal{I}_c$ implies the feasibility of unit rate with error probability $6k\epsilon$ in $\mathcal{I}$.
By setting $\epsilon=0$, we generalize the result in \cite{Huang:2014vg} regarding the zero-error model as a special case.
\begin{cor}\label{corzerr}
Given any multiple-unicast network coding problem $\mathcal{I}$ with source-destination pairs $\{ (s_i, t_i), i=1,...,k \}$, a corresponding single-source single-terminal network error correction problem $\mathcal{I}_c=(\mathcal{G},s,t,\mathcal{A})$  in which $\mathcal{A}$ includes sets with at most a single edge can be constructed as specified in Figure \ref{zeroerr},  such that rate $k$ is feasible with zero error in $\mathcal{I}_c$ if and only if unit rate is feasible with zero error in $\mathcal{I}$.
\end{cor}

\section{Unachievability of Network Error Correction Capacity}
Theorem \ref{epsilonerrk2} and Corollary \ref{corzerr} suggest that there are strong ``if and only if'' connections between $\mathcal{I}_c$ and $\mathcal{I}$ for both zero-error feasibility and vanishing-error feasibility.
It is natural to ask if this connection extends to the case of asymptotic feasibility, i.e., if it is true that rate $k$ is asymptotically feasible in $\mathcal{I}_c$ if and only if unit rate is asymptotically feasible in $\mathcal{I}$. We answer this question negatively by constructing a counter-example. 
\begin{figure}[h!]
  \begin{center}
      \includegraphics[width=0.26\textwidth]{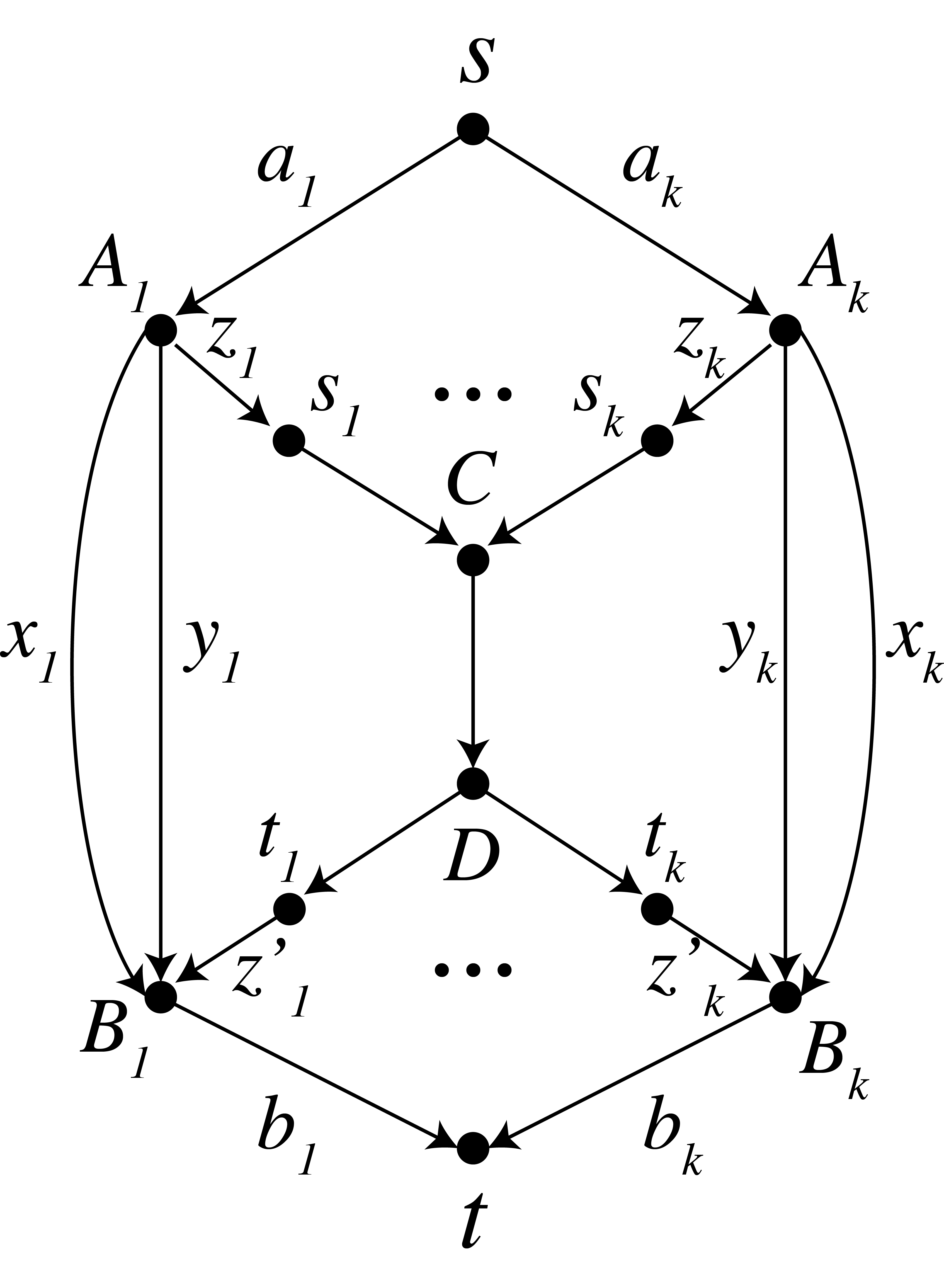}
  \caption{Construction of $\mathcal{I}_c$ and $\mathcal{I}$. In $\mathcal{I}_c$, the source is $s$ and the terminal is $t$. $\mathcal{A}$ includes all singleton sets of a single edge except $\{a_i\}$ and $\{b_i\}$, $i=1,...,k$. In $\mathcal{I}$, the  source-destination pairs are $(s_i,t_i), i=1,...,k$. All edges have unit capacity.
  }\label{nonach}
           \end{center}
\end{figure}
\begin{theorem}\label{nonachex}
There exists a multiple-unicast network coding problem $\mathcal{I}$ and a corresponding single-source single-terminal network error correction problem $\mathcal{I}_c$, constructed from $\mathcal{I}$ as specified in Figure \ref{nonach}, such that rate $k$ is asymptotically feasible in $\mathcal{I}_c$, but unit rate is not asymptotically feasible in $\mathcal{I}$.
\end{theorem}
\begin{proof}
The construction of $\mathcal{I}$ and the corresponding $\mathcal{I}_c$ are shown in Figure \ref{nonach}. By the cut-set bounds,  any rate larger than $1/k$ is not feasible in $\mathcal{I}$. This shows the second statement of the theorem. 

We prove the first statement of the theorem by describing a network code with length $n$ that achieves rate $k - k/n$ in $\mathcal{I}_c$. First divide the source  message of rate $k-k/n$ into $k$ pieces $M=(M_1,...,M_k)$, such that $M_i$, $i=1,...,k$ are i.i.d. uniformly distributed over $[2^{n-1}]$. We denote $\phi_e : [2^{k(n-1)}] \to [2^n]$ as the encoding function\footnote{This is called the \emph{global encoding fuction} in the context of network coding.} of edge $e$, which takes the source message $M$ as input, and outputs the signal to be transmitted on $e$ when there is no error in the network. For all $i = 1,...,k$, we let
\begin{align*}
	\phi_{a_i}(M) = \phi_{x_i}(M) = \phi_{y_i}(M) = \phi_{z_i}(M) = \phi_{(s_i,C)}(M) = M_i
\end{align*}
Furthermore, we let
\begin{align*}
\phi_{(C,D)}(M) =  \phi_{(D,t_i)}(M) = \phi_{z'_i} (M)= \sum_{j=1}^k M_j, \ \ \forall i =1,...,k
\end{align*}
where the summation is bitwise xor. Note that the edges $a_i,x_i,y_i,z_i,(s_i,C),(C,D),(D,t_i),z'_i$ each has a capacity to transmit $n$ bits. But we only require each of them to transmit $n-1$ bits. Hence each edge reserves one unused bit. 

Node $B_i$, by observing the (possibly corrupted) signals received from edges $x_i, y_i, z'_i$, performs error correction in the following way. If the signal (of $n-1$ bits) received from $x_i$ equals the signal received from $y_i$, forward the signal to edge $b_i$, and then transmit one bit of 0 using the reserved bit.
Otherwise, forward the signal received from $z'_i$ to  $b_i$, and then transmit one bit of 1 using the reserved bit. 

Finally, terminal $t$ recovers the source message in the following way. For $i=1,...,k$, such that the reserved bit on $b_i$ equals $0$, decode the remaining $n-1$ bits received from $b_i$ as $\hat{M}_i$. If $\hat{M}_1,...,\hat{M}_k$ are all obtained in this way then decoding is completed.  If two or more pieces of the $\hat{M}_i$'s are not obtained, then a decoding failure is declared. Otherwise, let $\hat{M}_l$ be the unique piece that is not obtained, then subtract $\sum_{j=1, j \ne l}^k \hat{M}_j$ from the signal received from $b_l$, and decode the result as $\hat{M}_l$.

It remains to be shown that $\hat{M}_i = M_i$, $\forall i=1,...,k$ regardless of the error patterns, and that a decoding failure will not be declared. Notice that the reserved bit on $b_i$ equals 1 only if an error occurs to either $x_i$ or $y_i$. Since there is at most a single error edge, for $i=1,...,k$ there is at most one $b_i$ with reserved bit 1. Therefore the decoder is able to obtain at least $k-1$ pieces of $\hat{M}_1, ..., \hat{M}_k$ during the first phase of decoding and will never declare failure. Next notice that if the reserved bit on $b_i$ is 0, then $\hat{M}_i \ne M_i$ only if errors occur to both $x_i$ and $y_i$. This is not possible by hypothesis and therefore $\hat{M}_i = M_i$ if the reserved bit on $b_i$ is 0. Finally, suppose the reserved bit on $b_i$ is 1, then an error must occur to either $x_i$ or $y_i$, and so $z'_i$ is not in error. Therefore $\hat{M}_i = \sum_{j=1}^k M_j  - \sum_{j=1, j \ne i}^k \hat{M}_j = M_i$. This proves the correctness of decoding and the second statement of the theorem.
\end{proof}

By applying the reduction result in Theorem \ref{epsilonerrk2} to the example constructed in Theorem \ref{nonachex}, it follows that the same example  shows the unachievability of single-unicast network error correction capacity. 
\begin{cor}\label{cor:2}
There exists a single-source single-terminal network error correction problem whose capacity is not feasible.
\end{cor}
\begin{proof}
The construction of the network error correction problem $\mathcal{I}_c$ is shown in Figure \ref{nonach}.  By the cut-set bounds, the capacity of $\mathcal{I}_c$ is upper bounded by $k$. By Theorem \ref{nonachex}, rate $k$ is asymptotically feasible in $\mathcal{I}_c$, and so the capacity of $\mathcal{I}_c$ is $k$. 
Also by Theorem \ref{nonachex}, unit rate in not feasible in $\mathcal{I}$, and so by Theorem \ref{epsilonerrk2}, rate $k$ is not feasible in $\mathcal{I}_c$. This shows that the capacity of $\mathcal{I}_c$ is not feasible. 
\end{proof}

Corollary \ref{cor:2} suggests that although the network error correction capacity is (by definition) asymptotically feasible,  in general it may not be exactly feasible. This is in contrast to the scenario of network error correction with uniform $\mathcal{A}$, i.e.,  $\mathcal{A}$ is the collection of all subsets containing $z$ channels. In this case the network capacity can be achieved by linear codes. Unachievability of capacity is also studied for multiple-unicast networks \cite{Dougherty:2006cg} and sum networks \cite{Rai:2012iq}. For both cases example networks for which the capacity is not achievable are constructed using matriod theory.

\section{Concluding Remarks}
We show that determining the feasibility of a rate tuple in a multiple-unicast network coding problem can be efficiently reduced to determining the feasibility of a corresponding rate in a corresponding single-unicast network error correction problem, where an adversary may jam at most a single edge. Note that though our  analysis assumes all source-destination pairs in the multiple-unicast transmit at equal rate, this restriction can be  relaxed by modeling a varying rate source $s$ as several equal rate sources co-located at $s$. Finally we apply the reduction to show the unachievability of single-unicast network error correction capacity. 

We note that our results do not imply that finding the capacity of a multiple-unicast network coding problem can be reduced to finding the capacity of a single-unicast network error correction problem.  Whether it is possible to construct such a reduction would be an interesting open problem. 

\newcommand{\BIBdecl}{\setlength{\itemsep}{ -0.15em}}
\bibliographystyle{IEEEtran}
\bibliography{ref}

\end{document}